\documentclass[a4paper,noarxiv,onecolumn]{quantumarticle}
\usepackage[utf8]{inputenc}
\usepackage[english]{babel}
\usepackage[T1]{fontenc}

\usepackage{amsmath,amsthm,amsfonts,amssymb,bm}
\usepackage{float}
\usepackage[numbers,sort&compress]{natbib}

\usepackage{lipsum}
\usepackage{mdframed}
\usepackage{booktabs}
\usepackage{longtable}

\newtheorem{theorem}{Theorem}

\usepackage{silence}
\begin{document}

\title{Multi-party dynamic quantum homomorphic encryption scheme based on rotation operators}
    
\author{Zhen-Zhen Li}
\email{zhenzhen\_li@bigc.edu.cn}
\affiliation{School of Information Engineering, Beijing Institute of Graphic Communication, Beijing 102600, China}
\author{Ming-Kui Liu}
\affiliation{School of Information Engineering, Beijing Institute of Graphic Communication, Beijing 102600, China}
\author{Wen-Ling Yang}
\affiliation{School of Information Engineering, Beijing Institute of Graphic Communication, Beijing 102600, China}
\author{Bo Gao}
\affiliation{School of Information Engineering, Beijing Institute of Graphic Communication, Beijing 102600, China}
\author{Zi-Chen Li}
\affiliation{School of Information Engineering, Beijing Institute of Graphic Communication, Beijing 102600, China}
   
\maketitle
	
\begin{abstract}
Quantum homomorphic encryption is the corresponding technology of classical homomorphic encryption in the quantum field. Due to its ability to ensure the correctness of computation and the security of data, it is particularly suitable for delegated computation in quantum cloud networks. However, previous schemes were unable to simultaneously handle the volatility problem of servers dynamically and eliminate the error caused by homomorphic evaluation of $T$-gate non-interactively. Therefore, a novel multi-party dynamic quantum homomorphic encryption scheme based on rotation operators is proposed in this paper. Firstly, the proposed scheme uses the rotation operators to solve the phase gate error that occurs during the homomorphic evaluation of $T$-gate non-interactively. Secondly, the scheme can dynamically deal with instability of servers, such as adding a server or removing a server. Then, the trusted key center is introduced, which is responsible for key updating and circuit replacement, which lowers the requirements for quantum capabilities on the client. Finally, this scheme extends the single-client multi-server model to the multi-client multi-server model, making it more suitable for quantum distributed networks and enhancing its practicality. In addition, we theoretically prove the correctness and fully homomorphic property of the proposed scheme, and verify it through the simulation conducted on the IBM Quantum Experience. Security analysis and efficiency analysis further demonstrate that the proposed scheme is information-theoretically secure and possesses high qubit efficiency.
\end{abstract}

\section{Introduction}\label{sec1}
With the advent of the big data era and the widespread adoption of cloud computing, data privacy and security have become focal points of societal concern. Ensuring the security of data during delegated computation is paramount. Is there a technology that can address this issue by enabling computation to be performed on encrypted data without decrypting? As early as 1978, Rivest et al. proposed the concept of homomorphic encryption, a term then referred to as privacy homomorphism \cite{1}. They envisioned an encryption system that would allow specific computation to be performed on encrypted data without needing to decrypt it. However, they did not propose a concrete implementation at that time. It wasn't until 2009 that Gentry introduced the first fully homomorphic encryption (FHE) scheme based on lattice theory \cite{2}, marking a significant breakthrough in homomorphic encryption research. This sparked interest in the international cryptographic community in fully homomorphic encryption schemes that enable arbitrary function transformations on ciphertexts.

However, most classical fully homomorphic encryption schemes rely on mathematical hard problems, mainly categorized into the approximate greatest common divisor problem on integers \cite{3,4,5} and the learning with errors problem based on lattices \cite{6,7,8}. With the development of quantum computing, traditional cryptography faces severe security challenges. For example, Shor's algorithm \cite{9} leverages the parallelism of quantum computing to factor large integers in polynomial time, thus threatening the security of classical encryption algorithms based on number theory, such as RSA and ECC. In contrast, quantum cryptography is based on quantum mechanics, with its security guaranteed by physical principles such as the uncertainty principle and the no-cloning theorem. In this context, quantum homomorphic encryption has emerged. Quantum homomorphic encryption is the quantum counterpart of classical homomorphic encryption, aiming to achieve secure and efficient homomorphic operations within the quantum computing framework. 

So far, quantum homomorphic encryption has been extensively studied. In 2012, Rohde et al. implemented a constrained quantum homomorphic scheme using the boson sampling model and quantum walks \cite{10}. In 2013, Liang firstly introduced the complete definitions of quantum homomorphic encryption (QHE) and quantum fully homomorphic encryption (QFHE), and constructed the first symmetric quantum homomorphic encryption scheme that showed clearly the framework of the scheme \cite{11}. In 2014, Yu et al. discovered a no-go result \cite{12}, stating that any perfectly information-theoretically secure QHE scheme would inevitably incur exponential overhead in a non-interactive setting. In the same year, Fisher et al. demonstrated arbitrary quantum computations on encrypted data using single photons and linear optics in experiments \cite{13}. In 2015, Broadbent and Jeffery (BJ15) constructed a QHE scheme based on classical fully homomorphic encryption, which could only perform a limited number of non-Clifford gates \cite{14}. This scheme possesses computational security and non-interactivity. In 2016, Dulek et al. extended the BJ15 scheme to develop a QHE scheme suitable for polynomial-sized non-Clifford gates \cite{15}. This scheme is not only compact but can also efficiently evaluate quantum circuits of arbitrary polynomial size. In 2017, Alagic et al. constructed a verifiable QFHE scheme \cite{16}, which provides authenticated encryption to ensure the correctness of homomorphic computation results. This scheme can perform arbitrary polynomial-time quantum computations without requiring interaction between the client and the server. In 2018, Ouyang et al. constructed an information-theoretically secure QHE scheme using quantum coding \cite{17}. This scheme is based on the structure of quantum codes and can only evaluate quantum circuits with a constant number of non-Clifford gates. In the same year, Mahadev developed the first classical-key FHE scheme for quantum circuits using a quantum capabilities approach \cite{18}. This scheme enables blind delegation of quantum computations to a trusted quantum server and malicious servers cannot gain any information from the process. Concurrently, Newman et al. \cite{19} and Lai et al. \cite{20} independently proved a strengthened no-go result: it is impossible to construct a non-interactive and information-theoretically secure quantum fully homomorphic encryption scheme. Therefore, non-interactive quantum fully homomorphic encryption schemes can only achieve computational security. In 2019, Chen et al. proposed a quantum homomorphic encryption scheme with a flexible number of servers based on the (k,n) threshold quantum secret sharing scheme \cite{21}. In 2020, Liang introduced the concepts of the encrypted gate and the gate teleportation based on the quantum teleportation. Using these concepts, Liang addressed the non-interactive evaluation of  $T$-gate, leading to the construction of two quasi-compact quantum homomorphic encryption schemes \cite{22}. In 2021, Zhang et al. proposed a probabilistic quantum homomorphic encryption scheme \cite{23}. This scheme utilizes pre-shared non-maximally entangled states between the client and server as auxiliary resources to probabilistically evaluate  $T$-gate, thereby reducing the requirement on auxiliary quantum states. In 2022, Zhang et al. extended their probabilistic quantum homomorphic encryption scheme to the multi-client scenario \cite{24}. In 2023, Chang et al. proposed a dynamic QFHE scheme which is suitable for universal quantum circuits \cite{25}. This scheme requires interaction between parties to correct errors that occur during  gate evaluations, and it addresses fluctuations in server performance. Also in 2023, Shang et al. introduced a two-round quantum homomorphic encryption scheme using matrix decomposition and circuit synthesis methods \cite{26}. They applied this scheme to the ciphertext retrieval experiment. In 2024, Pan et al. were the first to introduce quantum network coding into quantum homomorphic encryption. They proposed a cross quantum homomorphic encryption scheme \cite{27} implemented in the butterfly network. However, schemes \cite{11,15,17,18,21,25} in quantum homomorphic encryption involve interaction between the client and the server during the evaluation phase, which consumes quantum resources and may even lead to potential key leakage during interaction. Although the schemes \cite{10,13,14,16,22,23,24,26} achieve non-interactive features using various methods, they do not address the issue of server fluctuations in delegated computation. For most of the above schemes, they operate in a single-server mode, which is not suitable for quantum distributed networks and lacks practicality.

Based on the analysis of the aforementioned schemes, a novel multi-party dynamic quantum homomorphic encryption scheme based on rotation operators is proposed in this paper, which can allow arbitrary quantum computations between multiple servers and multiple clients. During the evaluation phase, no interaction is required between the parties.

Our main contributions are as follows:
\begin{itemize}
\item[$\bullet$] A novel dynamic quantum homomorphic encryption scheme based on rotation operators is proposed. The scheme can simultaneously handle the volatility problem of the servers dynamically and achieve non-interactive evaluation of $T$-gate using rotation operators. 
\item[$\bullet$] The trusted key center is introduced, which is responsible for key updating and circuit replacement, which reduces the workload of the client and even lowers the requirements for quantum capabilities on the client.
\item[$\bullet$] The proposed scheme extends the single-client multi-server model to the multi-client multi-server model, making it more suitable for quantum distributed networks and enhancing its practicality. 
\item[$\bullet$] We conduct the simulation on IBM Quantum Experience to validate the correctness of the scheme. Through theoretical analysis, we demonstrate that the scheme achieves information-theoretic security, full homomorphism and high qubit efficiency.
\end{itemize}

The remainders of this paper are organized as follows. In Sect.\ref{sec2}, we introduce the relevant preliminaries, including quantum computing, GHZ measurement, quantum gate replacement, and quantum homomorphic encryption. In Sect.\ref{sec3}, a new multi-party dynamic quantum homomorphic encryption scheme based on rotation operators is proposed. In Sect.\ref{sec4}, we provide a detailed analysis from the perspectives of correctness, security, and quantum bit efficiency, and we validate the scheme through the simulation on IBM Quantum Experience. Finally, the conclusion is presented in Sect.\ref{sec5}.

\section{Preliminaries}\label{sec2}

\subsection{Quantum computing}\label{subsec2.1}

Here, we briefly describe some common quantum gates. For details, please refer to \cite{28}. Quantum circuits used in quantum computing consist of Clifford gates and non-Clifford gates, i.e. $G = \left\{ {X,Z,H,S,T,CNOT} \right\}$. The corresponding unitary matrices are defined as follows:

\begin{equation}\label{eq1}
X = \left[ {\begin{array}{*{20}{c}}
0&1\\
1&0
\end{array}} \right]Z = \left[ {\begin{array}{*{20}{c}}
1&0\\
0&{ - 1}
\end{array}} \right]H = \frac{1}{{\sqrt 2 }}\left[ {\begin{array}{*{20}{c}}
1&1\\
1&{ - 1}
\end{array}} \right]S = \left[ {\begin{array}{*{20}{c}}
1&0\\
0&i
\end{array}} \right]T = \left[ {\begin{array}{*{20}{c}}
1&0\\
0&{{e^{{{i\pi } \mathord{\left/
 {\vphantom {{i\pi } 4}} \right.
 \kern-\nulldelimiterspace} 4}}}}
\end{array}} \right]
\end{equation}

\begin{equation}\label{eq2}
CNOT = \left[ {\begin{array}{*{20}{c}}
1&0&0&0\\
0&1&0&0\\
0&0&0&1\\
0&0&1&0
\end{array}} \right]
\end{equation}

\subsection{GHZ state and the properties}\label{subsec2.2}

The GHZ state is an entangled state of n-particle, exhibiting maximum entanglement. According to transformation formulas, its eigenstates for three-particle are given as follows:

\begin{equation}\label{eq3}
\left| {{\Psi _1}} \right\rangle  = \frac{1}{{\sqrt 2 }}\left( {\left| {000} \right\rangle  + \left| {111} \right\rangle } \right) = \frac{1}{2}\left( {\left| { +  +  + } \right\rangle  + \left| { +  -  - } \right\rangle  + \left| { -  +  - } \right\rangle  + \left| { -  -  + } \right\rangle } \right)
\end{equation}

\begin{equation}\label{eq4}
\left| {{\Psi _2}} \right\rangle  = \frac{1}{{\sqrt 2 }}\left( {\left| {000} \right\rangle  - \left| {111} \right\rangle } \right) = \frac{1}{2}\left( {\left| { +  +  - } \right\rangle  + \left| { +  -  + } \right\rangle  + \left| { -  +  + } \right\rangle  + \left| { -  -  - } \right\rangle } \right)
\end{equation}

\begin{equation}\label{eq5}
\left| {{\Psi _3}} \right\rangle  = \frac{1}{{\sqrt 2 }}\left( {\left| {100} \right\rangle  + \left| {011} \right\rangle } \right) = \frac{1}{2}\left( {\left| { +  +  + } \right\rangle  + \left| { +  -  - } \right\rangle  - \left| { -  +  - } \right\rangle  - \left| { -  -  + } \right\rangle } \right)
\end{equation}

\begin{equation}\label{eq6}
\left| {{\Psi _4}} \right\rangle  = \frac{1}{{\sqrt 2 }}\left( {\left| {100} \right\rangle  - \left| {011} \right\rangle } \right) = \frac{1}{2}\left( {\left| { +  +  - } \right\rangle  + \left| { +  -  + } \right\rangle  - \left| { -  +  + } \right\rangle  - \left| { -  -  - } \right\rangle } \right)
\end{equation}

\begin{equation}\label{eq7}
\left| {{\Psi _5}} \right\rangle  = \frac{1}{{\sqrt 2 }}\left( {\left| {010} \right\rangle  + \left| {101} \right\rangle } \right) = \frac{1}{2}\left( {\left| { +  +  + } \right\rangle  - \left| { +  -  - } \right\rangle  + \left| { -  +  - } \right\rangle  - \left| { -  -  + } \right\rangle } \right)
\end{equation}

\begin{equation}\label{eq8}
\left| {{\Psi _6}} \right\rangle  = \frac{1}{{\sqrt 2 }}\left( {\left| {010} \right\rangle  - \left| {101} \right\rangle } \right) = \frac{1}{2}\left( {\left| { +  +  - } \right\rangle  - \left| { +  -  + } \right\rangle  - \left| { -  -  + } \right\rangle  - \left| { -  -  - } \right\rangle } \right)
\end{equation}

\begin{equation}\label{eq9}
\left| {{\Psi _7}} \right\rangle  = \frac{1}{{\sqrt 2 }}\left( {\left| {110} \right\rangle  + \left| {001} \right\rangle } \right) = \frac{1}{2}\left( {\left| { +  +  + } \right\rangle  - \left| { +  -  - } \right\rangle  - \left| { -  +  - } \right\rangle  + \left| { -  -  + } \right\rangle } \right)
\end{equation}

\begin{equation}\label{eq10}
\left| {{\Psi _8}} \right\rangle  = \frac{1}{{\sqrt 2 }}\left( {\left| {110} \right\rangle  - \left| {001} \right\rangle } \right) = \frac{1}{2}\left( {\left| { +  +  - } \right\rangle  - \left| { +  -  + } \right\rangle  - \left| { -  +  + } \right\rangle  + \left| { -  -  - } \right\rangle } \right)
\end{equation}
where $\left|  +  \right\rangle $ and $\left|  -  \right\rangle $ are defined such that $\left|  +  \right\rangle  = \frac{1}{{\sqrt 2 }}\left( {\left| 0 \right\rangle  + \left| 1 \right\rangle } \right)$ and $\left|  -  \right\rangle  = \frac{1}{{\sqrt 2 }}\left( {\left| 0 \right\rangle  - \left| 1 \right\rangle } \right)$, and their combined set $\left\{ {\left|  +  \right\rangle ,\left|  -  \right\rangle } \right\}$ is referred to as the Hadamard basis, also known as the X-basis.

Suppose we prepare three Bell states from ${\beta _{00}} = {1 \mathord{\left/
 {\vphantom {1 {\sqrt 2 }}} \right.
 \kern-\nulldelimiterspace} {\sqrt 2 }}\left( {\left| {00} \right\rangle  + \left| {11} \right\rangle } \right)$
 and perform the joint GHZ measurement on the first particle of each pair of Bell state. According to the theory of entanglement exchange, at the same time the second particle of each pair of Bell state will also be entangled into the same GHZ state. The specific process is shown in Equation (\ref{eq11}).

\begin{equation}\label{eq11}
\begin{gathered}
  {\left| {{\beta _{00}}} \right\rangle _{12}}{\left| {{\beta _{00}}} \right\rangle _{34}}{\left| {{\beta _{00}}} \right\rangle _{56}} \hfill \\
   = \frac{1}{{\sqrt 8 }}\left( \begin{gathered}
  {\left| {000} \right\rangle _{135}}{\left| {000} \right\rangle _{246}} + {\left| {001} \right\rangle _{135}}{\left| {001} \right\rangle _{246}} + {\left| {010} \right\rangle _{135}}{\left| {010} \right\rangle _{246}} + {\left| {011} \right\rangle _{135}}{\left| {011} \right\rangle _{246}} \hfill \\
   + {\left| {100} \right\rangle _{135}}{\left| {100} \right\rangle _{246}} + {\left| {101} \right\rangle _{135}}{\left| {101} \right\rangle _{246}} + {\left| {110} \right\rangle _{135}}{\left| {110} \right\rangle _{246}} + {\left| {111} \right\rangle _{135}}{\left| {111} \right\rangle _{246}} \hfill \\ 
\end{gathered}  \right) \hfill \\
   = \frac{1}{{\sqrt 8 }}\left( \begin{gathered}
  {\left| {{\Psi _1}} \right\rangle _{135}}{\left| {{\Psi _1}} \right\rangle _{246}} + {\left| {{\Psi _2}} \right\rangle _{135}}{\left| {{\Psi _2}} \right\rangle _{246}} + {\left| {{\Psi _3}} \right\rangle _{135}}{\left| {{\Psi _3}} \right\rangle _{246}} + {\left| {{\Psi _4}} \right\rangle _{135}}{\left| {{\Psi _4}} \right\rangle _{246}} \hfill \\
   + {\left| {{\Psi _5}} \right\rangle _{135}}{\left| {{\Psi _5}} \right\rangle _{246}} + {\left| {{\Psi _6}} \right\rangle _{135}}{\left| {{\Psi _6}} \right\rangle _{246}} + {\left| {{\Psi _7}} \right\rangle _{135}}{\left| {{\Psi _7}} \right\rangle _{246}} + {\left| {{\Psi _8}} \right\rangle _{135}}{\left| {{\Psi _8}} \right\rangle _{246}} \hfill \\ 
\end{gathered}  \right) \hfill \\ 
\end{gathered} 
\end{equation}
where the subscript numbers represent the identification numbers of particles within the six-particle system.

After performing the single-qubit measurement in the X basis on the obtained GHZ state, we can obtain the measurement result $measure = \left\{ {M{R_1},M{R_2},M{R_3}} \right\}$. We interpret $\left|  +  \right\rangle $ as "0" and $\left|  -  \right\rangle $ as "1". Therefore, based on the chosen one of the eight eigenstates, it can be deduced that the secret to be shared is $K = M{R_1} \oplus M{R_2} \oplus M{R_3}$.

Based on the above forms of the three-particle GHZ state eigenstates and their measurement properties, the corresponding shared secrets for GHZ measurements are shown in the Table \ref{Table-1}.

\begin{table}[htbp]
\begin{center}
 \caption{The corresponding shared secret for the GHZ basis selected by the client}
  \label{Table-1}
  \begin{tabular}{@{}cc@{}}
    \hline
    The chosen GHZ basis &   The corresponding shared secret \textit{K} \\
    \hline
$\left| {{\Psi _1}} \right\rangle $  & 0\\
$\left| {{\Psi _2}} \right\rangle $  & 1\\
$\left| {{\Psi _3}} \right\rangle $  & 0\\
$\left| {{\Psi _4}} \right\rangle $  & 1\\
$\left| {{\Psi _5}} \right\rangle $  & 0\\
$\left| {{\Psi _6}} \right\rangle $  & 1\\
$\left| {{\Psi _7}} \right\rangle $  & 0\\
$\left| {{\Psi _8}} \right\rangle $  & 1\\
    \hline
  \end{tabular}
\end{center}
\end{table}

The aforementioned properties can be generalized. The n-particle GHZ state is given by

\begin{equation}\label{eq12}
\left| {{G_n}} \right\rangle  = \frac{1}{{\sqrt 2 }}\left( {\left| {{p_1},{p_2}, \cdots ,{p_n}} \right\rangle  \pm \left| {\overline {{p_1}} ,\overline {{p_2}} , \cdots ,\overline {{p_n}} } \right\rangle } \right)
\end{equation}
where $\overline {{p_i}} $ is the bitwise negation of ${p_i}\left( {{p_i} \in \left\{ {0,1} \right\},i = 1,2, \cdots n} \right)$.

For the above process, it goes through three steps: (i) Prepare \textit{n} pairs of Bell states, (ii) Conduct n-particle GHZ measurements, (iii) Conduct single-particle X-basis measurements. After that, the n-particle GHZ state has all the above properties. That is to say, the three-particle GHZ state is the case where the value of \textit{n} for the n-particle GHZ state is 3.

\subsection{Quantum gate replacement}\label{subsec2.3}

In quantum computing, quantum circuits are composed of Clifford gates and non-Clifford gates. Clifford gates are characterized by their commutativity with Pauli operators, ensuring that no errors occur during the evaluation phase. As is well known, $T$-gate is a non-Clifford gate. When applied to an encrypted state during the evaluation phase, it has a specific formula $T{X^a}{Z^b}\left| \varphi  \right\rangle  = {X^{^a}}{Z^{a \oplus b}}{S^a}T\left| \varphi  \right\rangle $. From the equation, it can be observed that an unexpected s-error appears in the homomorphic result, leading to incorrect ciphertext outputs from the server's evaluation. Therefore, efforts must be made to eliminate s-error to ensure the correctness of the homomorphism. Here, we utilize the properties of rotation operators to replace circuits before the evaluation phase, applying the modified circuits to encrypted states, and finally decrypting to obtain the correct homomorphic results. Next, we will elaborate on the main idea of this approach \cite{29}.

First, we introduce a theorem \cite{28} and its generalized form. Let $U$ be a unitary operator on a single qubit. Then there exist real numbers $\alpha $, $\beta $, $\gamma $ and $\delta $ such that

\begin{equation}\label{eq13}
U = {e^{^{i\alpha }}}{R_z}\left( \beta  \right){R_y}\left( \gamma  \right){R_z}\left( \delta  \right) \buildrel \Delta \over = U\left( {\alpha ,\beta ,\gamma ,\delta } \right)
\end{equation}
Regarding the rotation operators around the Z and Y axes, the following formulas hold true:

\begin{equation}\label{eq14}
{R_z}\left( {{{\left( { - 1} \right)}^a}\theta } \right){X^a}{Z^b} = {X^a}{Z^b}{R_z}\left( \theta  \right)
\end{equation}

\begin{equation}\label{eq15}
{R_y}\left( {{{\left( { - 1} \right)}^{a + b}}\theta } \right){X^a}{Z^b} = {X^a}{Z^b}{R_y}\left( \theta  \right)
\end{equation}
From Equation (\ref{eq13},\ref{eq14},\ref{eq15}) and the theorem \cite{28}, we have

\begin{equation}\label{eq16}
\begin{array}{l}
{X^a}{Z^b}U\left( {\alpha ,\beta ,\gamma ,\delta } \right)\\
 = {e^{^{i\alpha }}}{X^a}{Z^b}{R_z}\left( \beta  \right){R_y}\left( \gamma  \right){R_z}\left( \delta  \right)\\
 = {e^{^{i\alpha }}}{R_z}\left( {{{\left( { - 1} \right)}^a}\beta } \right){R_y}\left( {{{\left( { - 1} \right)}^{a + b}}\gamma } \right){R_z}\left( {{{\left( { - 1} \right)}^a}\delta } \right){X^a}{Z^b}\\
 = U\left( {\alpha ,{{\left( { - 1} \right)}^a}\beta ,{{\left( { - 1} \right)}^{a + b}}\gamma ,{{\left( { - 1} \right)}^a}\delta } \right){X^a}{Z^b}
\end{array}
\end{equation}

Next, we proceed with the replacement of quantum gates. According to the theorem, the $T$-gate can be decomposed as $T = U\left( {{\pi  \mathord{\left/
 {\vphantom {\pi  8}} \right.
 \kern-\nulldelimiterspace} 8},\beta ,0,\delta } \right)$
, where $\beta  + \delta  = {\pi  \mathord{\left/
 {\vphantom {\pi  4}} \right.
 \kern-\nulldelimiterspace} 4}$. Consequently, based on the encryption key $\left( {a,b} \right)$ and Equation (\ref{eq16}),  $T$-gate in the quantum circuit is replaced by $T' = U\left( {{\pi  \mathord{\left/
 {\vphantom {\pi  8}} \right.
 \kern-\nulldelimiterspace} 8},{{\left( { - 1} \right)}^a}\beta ,0,{{\left( { - 1} \right)}^a}\delta } \right)$. However, such a straightforward replacement of  $T$-gate in circuits may potentially leak the key $a$. The detailed proof can be found in \cite{29}.

Finally, an improved replacement method is needed. Note that

\begin{equation}\label{eq17}
T = \left[ {\begin{array}{*{20}{c}}
1&0\\
0&{{e^{{{i\pi } \mathord{\left/
 {\vphantom {{i\pi } 4}} \right.
 \kern-\nulldelimiterspace} 4}}}}
\end{array}} \right] = {e^{{{i\pi } \mathord{\left/
 {\vphantom {{i\pi } 8}} \right.
 \kern-\nulldelimiterspace} 8}}}\left[ {\begin{array}{*{20}{c}}
{{e^{ - {{i\pi } \mathord{\left/
 {\vphantom {{i\pi } 8}} \right.
 \kern-\nulldelimiterspace} 8}}}}&0\\
0&{{e^{{{i\pi } \mathord{\left/
 {\vphantom {{i\pi } 8}} \right.
 \kern-\nulldelimiterspace} 8}}}}
\end{array}} \right] = {e^{{{i\pi } \mathord{\left/
 {\vphantom {{i\pi } 8}} \right.
 \kern-\nulldelimiterspace} 8}}}{R_z}\left( {\frac{\pi }{4}} \right)
\end{equation}

Since the global phase ${e^{{\pi  \mathord{\left/
 {\vphantom {\pi  8}} \right.
 \kern-\nulldelimiterspace} 8}}}$
 of the $T$-gate does not affect measurement results, $T$-gate can be replaced with ${R_z}\left( {{\pi  \mathord{\left/
 {\vphantom {\pi  4}} \right.
 \kern-\nulldelimiterspace} 4}} \right)$, and ${T^\dag }$-gate can be replaced with ${R_z}\left( { - {\pi  \mathord{\left/
 {\vphantom {\pi  4}} \right.
 \kern-\nulldelimiterspace} 4}} \right)$. Thus, based on the encryption key $\left( {a,b} \right)$, the replacement rules can be improved to:

\begin{equation}\label{eq18}
\begin{array}{l}
T \to {R_z}\left( {{{\left( { - 1} \right)}^a}{\pi  \mathord{\left/
 {\vphantom {\pi  4}} \right.
 \kern-\nulldelimiterspace} 4}} \right)\\
{T^\dag } \to {R_z}\left( {{{\left( { - 1} \right)}^a}\left( {{{ - \pi } \mathord{\left/
 {\vphantom {{ - \pi } 4}} \right.
 \kern-\nulldelimiterspace} 4}} \right)} \right)
\end{array}
\end{equation}

\subsection{Quantum homomorphic encryption}\label{subsec2.4}

In 2013, Liang \cite{11} formally defined quantum homomorphic encryption and proposed the first symmetric quantum homomorphic encryption scheme that showed clearly the framework of the scheme.

Quantum homomorphic encryption schemes consist of the following four algorithms: key generation, encryption, evaluation and decryption.

(i) \textit{Key generation}: This algorithm generates two types of key. The first type is the encryption key ${e_k}$, and the other is the evaluation key ${\rho _{evk}}$. According to the key update algorithm, the decryption key ${d_k}$ can be computed using the encryption key ${e_k}$.

(ii) \textit{Encryption}: We use an algorithm $Enc$ as the encryption algorithm, which operates on the quantum states $\rho $ to produce the encrypted state $\sigma  = Enc\left( {{e_k},\rho } \right)$.

(iii) \textit{Evaluation}: We use an algorithm $Eval$ as the evaluation algorithm, which performs computations ${C_q}$ on the encrypted state $\sigma $ and outputs the evaluated ciphertext $\sigma ' = Eva{l_{{\rho _{evk}}}}\left( {{C_q},\sigma } \right)$.

(iv) \textit{Decryption}: We use an algorithm $Dec$ as the decryption algorithm. Using the decryption key ${d_k}$, it operates on the encrypted state $\sigma '$ and outputs the evaluated plaintext state $\rho ' = Dec\left( {{d_k},\sigma '} \right)$.

\section{Multi-party dynamic quantum homomorphic encryption scheme based on rotation operators}\label{sec3}

In this chapter, a novel multi-party dynamic quantum homomorphic encryption (MDQHE) scheme based on rotation operators is proposed. This scheme allows arbitrary quantum computations between multiple servers and multiple clients. During the evaluation phase, no interaction is required between the parties. In Sect.\ref{subsec3.1}, we construct the (M+N)-party dynamic quantum homomorphic encryption scheme. Then, based on the scheme in Sect.\ref{subsec3.1}, we describe the dynamic scenarios of adding a server in Sect.\ref{subsec3.2} and removing a server in Sect.\ref{subsec3.3}. 

Specifically, the client needs to complete a massive computational task but lacks the capability to do so independently, thus requiring delegation to multiple servers. If there are multiple clients, the delegated computation can be executed in parallel. These servers possess powerful computational abilities to meet user's demands, but they may be dishonest and steal the user's private data potentially. A trusted key center is introduced, which is fully trustworthy and does not disclose any secret data to anyone. It is responsible for generating secure keys and updating keys securely, as well as conducting quantum circuit replacements for the homomorphic evaluation phase, specifically referring to $T/{T^\dag }$-gate replacement. This reduces the client's operations and lowers the requirements for the client's quantum capabilities.

\subsection{Dynamic quantum homomorphic encryption scheme with (M+N)-party}\label{subsec3.1}

In this section, (M+N)-party dynamic quantum homomorphic encryption scheme is constructed. The proposed scheme allows \textit{N} clients to simultaneously delegate computations on encrypted data to \textit{M} servers, enabling tasks that clients cannot perform themselves. There is no need for interaction between the parties during the homomorphic evaluation process. Additionally, QOTP is used to ensure that secret data remains confidential and protected against leakage.

Suppose quantum universal circuit is composed of gates from a set of quantum gates $G = \left\{ {X,Z,H,S,T,CNOT} \right\}$. Assume there are \textit{N} clients, represented as $clien{t_i}$, where $i = 1,2,\, \cdots ,N$, and \textit{M} servers, represented as $serve{r_j}$, where $j = 1,2, \cdots ,M$. In the secret splitting phase, clients prepare Bell states and perform the GHZ measurement and X-basis measurement on corresponding particles to achieve to split secret. In the key generation phase, servers use measurement-device-independent quantum key distribution (MDI-QKD) protocol to distribute secure keys among multiple clients and the trusted key center. Simultaneously, the trusted key center replaces  $T$-gate in the circuit according to the keys and substitution rules, then sends the substituted quantum circuit to the servers. In the encryption phase, each client encrypts the secret data using quantum one-time pad (QOTP) technology and sends the ciphertext states to the corresponding server. In the evaluation phase, servers perform evaluations on the ciphertext states based on the sequence of quantum gates in the quantum circuit. Especially when applying the  $T$-gate, remember to use the substituted quantum circuit. After completing all evaluations, the corresponding ciphertext states are sent back to the clients. In the decryption phase, the trusted key center updates keys according to key update rules to generate decryption keys and sends them to the clients. Subsequently, clients use these decryption keys to decrypt and obtain the plaintext state that results from the homomorphic computation. These processes are illustrated in Fig.\ref{Fig-1}.

\begin{figure}[H]
\centering		 
\includegraphics[width=0.9\textwidth]{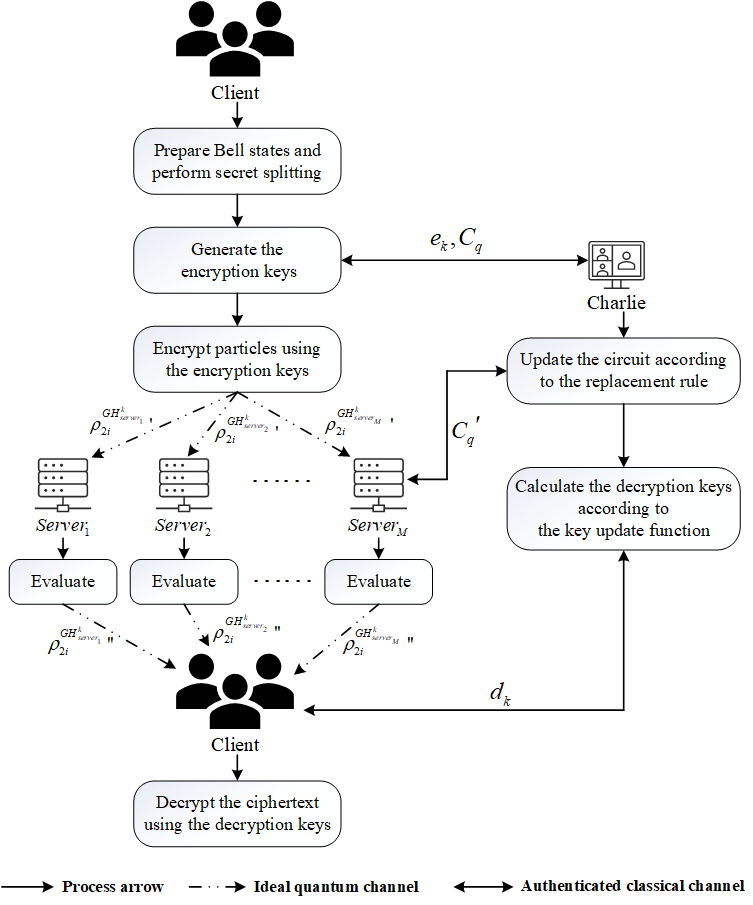}
\caption{\label{Fig-1} Flow chart of quantum homomorphic encryption scheme}
\end{figure}

\sloppy
\subsubsection{Secret Splitting}\label{subsec3.1.1}

Each client $clien{t_i}\left( {i = 1,2,\, \cdots ,N} \right)$ prepares \textit{nM} Bell states in ${\beta _{00}} = {1 \mathord{\left/
 {\vphantom {1 {\sqrt 2 }}} \right.
 \kern-\nulldelimiterspace} {\sqrt 2 }}\left( {\left| {00} \right\rangle  + \left| {11} \right\rangle } \right)$. Then, the client $clien{t_i}\left( {i = 1,2,\, \cdots ,N} \right)$ divides the particles in the Bell states into \textit{M} parts, namely: $GH_{serve{r_1}}^i = \left\{ {\rho _{1i}^{GH_{serve{r_1}}^k},\rho _{2i}^{GH_{serve{r_1}}^k}} \right\}$, $GH_{serve{r_2}}^i = \left\{ {\rho _{1i}^{GH_{serve{r_2}}^k},\rho _{2i}^{GH_{serve{r_2}}^k}} \right\}$, $ \cdots $, $GH_{serve{r_M}}^i = \left\{ {\rho _{1i}^{GH_{serve{r_M}}^k},\rho _{2i}^{GH_{serve{r_M}}^k}} \right\}$, where $k = 1,2, \cdots ,n$, $i = 1,2,\, \cdots ,N$, \textit{i} denotes the i-th client, ${\rho _{1i}}$ and ${\rho _{2i}}$ represent the first and second particles of the Bell states prepared by the i-th client, respectively. Then, the client performs the joint n-particle GHZ measurement on the first particle of each pair of Bell states. According to the properties of the GHZ state in Sect.\ref{subsec2.2}, the second particle of each pair of Bell states will simultaneously become entangled into the same GHZ state, forming two particle sets, namely: $GH_1^i = \left\{ {\rho _{1i}^{GH_{serve{r_1}}^k},\rho _{1i}^{GH_{serve{r_2}}^k}, \cdots ,\rho _{1i}^{GH_{serve{r_M}}^k}} \right\}$, $GH_2^i = \left\{ {\rho _{2i}^{GH_{serve{r_1}}^k},\rho _{2i}^{GH_{serve{r_2}}^k}, \cdots ,\rho _{2i}^{GH_{serve{r_M}}^k}} \right\}$. Subsequently, by performing the single-particle X-basis measurement on the two particle sets, the measurement results can be obtained, namely: $measur{e_{GH_1^i}} = \left\{ {\rho _{1i}^{MR_{serve{r_1}}^k},\rho _{1i}^{MR_{serve{r_2}}^k}, \cdots ,\rho _{1i}^{MR_{serve{r_M}}^k}} \right\}$, $measur{e_{GH_2^i}} = \left\{ {\rho _{2i}^{MR_{serve{r_1}}^k},\rho _{2i}^{MR_{serve{r_2}}^k}, \cdots ,\rho _{2i}^{MR_{serve{r_M}}^k}} \right\}$. Based on the basis chosen during the GHZ measurement process, the secret that the i-th client $clien{t_i}\left( {i = 1,2,\, \cdots ,N} \right)$ intends to share can be easily deduced as:

\begin{equation}\label{eq19}
\begin{array}{l}
{K_i} = \rho _{1i}^{M{R_{serve{r_1}}}} \oplus \rho _{1i}^{M{R_{serve{r_2}}}} \oplus  \cdots  \oplus \rho _{1i}^{M{R_{serve{r_M}}}}\\
{\rm{    }} = \rho _{2i}^{M{R_{serve{r_1}}}} \oplus \rho _{2i}^{M{R_{serve{r_2}}}} \oplus  \cdots  \oplus \rho _{2i}^{M{R_{serve{r_M}}}}
\end{array}
\end{equation}

For each \textit{k}, corresponding to each bit of secret data, we have the following formula:

\begin{equation}\label{eq20}
\begin{array}{l}
{K_i}\left( k \right) = \rho _{1i}^{MR_{serve{r_1}}^k} \oplus \rho _{1i}^{MR_{serve{r_2}}^k} \oplus  \cdots  \oplus \rho _{1i}^{MR_{serve{r_M}}^k}\\
{\rm{         }} = \rho _{2i}^{MR_{serve{r_1}}^k} \oplus \rho _{2i}^{MR_{serve{r_2}}^k} \oplus  \cdots  \oplus \rho _{2i}^{MR_{serve{r_M}}^k}
\end{array}
\end{equation}

For each \textit{i}, corresponding to each client, the process remains the same.

To clearly understand the complex process of secret splitting, we will give an example. Here, assume that the number of clients is one, i.e. $N = 1$, the number of servers is three, i.e. $M = 3$, and the secret data is five bits, i.e $n = 5$. The client $clien{t_1}$ needs to prepare fifteen pairs of Bell states in ${\beta _{00}} = {1 \mathord{\left/
 {\vphantom {1 {\sqrt 2 }}} \right.
 \kern-\nulldelimiterspace} {\sqrt 2 }}\left( {\left| {00} \right\rangle  + \left| {11} \right\rangle } \right)$, then divides them into three groups, with five pairs of Bell states in each group. The joint five-particle GHZ measurement is performed on the first particle of the Bell states within each group. According to the properties of the GHZ state in Sect.\ref{subsec2.2}, the second particles of the Bell states within each group will also entangle into the same GHZ state as the first particles of the Bell states. Then, the X-basis measurement is performed on the second particle of the Bell states within each group respectively. The specific splitting process is shown in the Fig.\ref{Fig-2}.

\begin{figure}[htbp]
\centering		 
\includegraphics[width=0.95\textwidth]{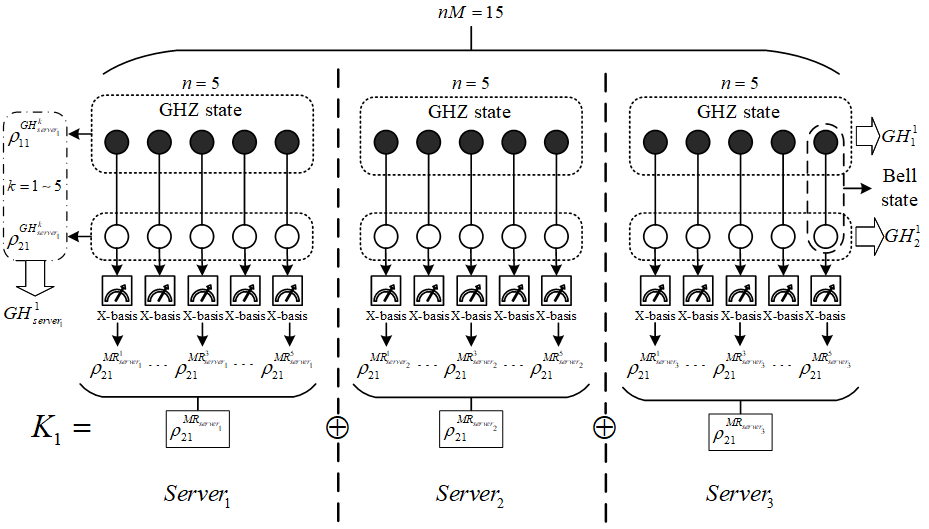}
\caption{\label{Fig-2} The example of secret splitting with one client, three servers and five-bit data}
\end{figure}

\subsubsection{Key Generation}\label{subsec3.1.2}

During this stage, two tasks need to be completed: servers participate in distributing keys among multiple clients and the trusted key center (i.e. Charlie), and the trusted key center performs substitution of the quantum circuit. The following two steps will be executed sequentially:

(i) \textit {Key Distribution}: Initially, the client $clien{t_i}\left( {i = 1,2,\, \cdots ,N} \right)$ and the trusted key center agree on a set of quantum states $\left\{ {\left| 1 \right\rangle ,\left| 0 \right\rangle ,\left|  +  \right\rangle ,\left|  -  \right\rangle } \right\}$ for key generation. Both parties randomly select states from the pre-agreed set. The l-th quantum state prepared by the i-th client is denoted as $\left| {{\varphi _{clien{t_{il}}}}} \right\rangle $, and the quantum states prepared by the trusted key center are denoted as $\left| {{\varphi _{Charli{e_{il}}}}} \right\rangle $. Next, the client and the trusted key center need to send the prepared quantum states to the servers. Upon receiving these states, the servers perform the joint Bell measurement on them, resulting in l pairs of entangled quantum states $\left| {{\varphi _{clien{t_{il}}}}} \right\rangle \left| {{\varphi _{Charli{e_{il}}}}} \right\rangle $. The servers then send the results to both parties through an authenticated classical channel. Based on the measurement results $\left\{ {{{\left| {{\Phi ^ \pm }} \right\rangle }_{clien{t_{il}}}}_{Charli{e_{il}}}} \right\}$, both parties need to retain the quantum states that correspond to the correctly measured results, and they need to inform each other about the bases used during the preparation of the quantum states. According to the BB84 protocol \cite{30}, only the quantum states where both parties used the same basis during state preparation are selected, termed as the sifted keys, which may constitute secure keys. The error rate of the selected measurement basis sequence is used to detect potential attacks. If the error rate is below the set threshold, error correction and privacy amplification are applied sequentially to the retained results, enhancing their security. Otherwise, if the error rate exceeds the set threshold, it indicates the presence of an eavesdropper, rendering the quantum channel insecure. In such cases, the protocol needs to be terminated and restarted. Finally, both parties encode the quantum states into binary bits according to pre-agreed rules:

\begin{equation}\label{eq21}
\begin{array}{l}
\left| 0 \right\rangle ,\left|  +  \right\rangle  \to 0\\
\left| 1 \right\rangle ,\left|  -  \right\rangle  \to 1
\end{array}
\end{equation}

Upon completing the above steps, both parties obtain the same secure key, represented as ${e_{{k_i}}} = \left( {{a_i},{b_i}} \right)$, where ${a_i},{b_i} \in {\left\{ {0,1} \right\}^{2n}}$. Since the encryption keys will be used to encrypt data after secret splitting, each client needs to generate \textit{M} pairs of keys ${e_{{k_i}}} = \left( {{a_i},{b_i}} \right)$ during the key generation process.

(ii) \textit {Circuit Substitution}: Firstly, the clients need to send the pre-prepared homomorphic evaluation circuit ${C_q}$ to the trusted key center, where the circuit ${C_q}$ is composed of Clifford gates and $T$-gate. Suppose the circuit ${C_q}$ consists of \textit{p} quantum gates, sequentially denoted as $Gate\left[ 1 \right]$, $Gate\left[ 2 \right]$, $ \cdots $, $Gate\left[ p \right]$, among which there are \textit{q} $T/{T^\dag }$-gate. For each $g \in \left\{ {1,2, \cdots ,p} \right\}$, two scenarios can occur:

If $Gate\left[ g \right] \notin \left\{ {T,{T^\dag }} \right\}$, then there is no need to replace the circuit.

If $Gate\left[ g \right] \in \left\{ {T,{T^\dag }} \right\}$, then it is required to update the $T/{T^\dag }$-gate according to the encryption key ${e_{{k_i}}} = \left( {{a_i},{b_i}} \right)$ and the substitution rules outlined in Sect.\ref{subsec2.3}. The replacement is as follows:

\begin{equation}\label{eq22}
\begin{gathered}
  T \to {R_z}\left( {{{\left( { - 1} \right)}^{{a_i}\left( \omega  \right)}}{\pi  \mathord{\left/
 {\vphantom {\pi  4}} \right.
 \kern-\nulldelimiterspace} 4}} \right) \hfill \\
  {T^\dag } \to {R_z}\left( {{{\left( { - 1} \right)}^{{a_i}\left( \omega  \right)}}\left( {{{ - \pi } \mathord{\left/
 {\vphantom {{ - \pi } 4}} \right.
 \kern-\nulldelimiterspace} 4}} \right)} \right) \hfill \\ 
\end{gathered}
\end{equation}
Where ${a_i}\left( \omega  \right)$ represents the partial encryption key corresponding to the $\omega \left( {\omega  = 1,2, \cdots ,n} \right)$-th qubit of the encryption key ${e_{{k_i}}} = \left( {{a_i},{b_i}} \right)$.

For the circuit ${C_q}$ that includes \textit{q} $T/{T^\dag }$-gate, the entire process requires \textit{q} rounds of quantum gate substitution.

After completing the above operations, the trusted key center sends the substituted circuit $C_q^{'}$ to the respective servers.

\sloppy
\subsubsection{Encryption}\label{subsec3.1.3}

Each client $clien{t_i}\left( {i = 1,2,\, \cdots ,N} \right)$ uses their respective encryption keys ${e_{{k_i}}} = \left( {{a_i},{b_i}} \right)$ and the QOTP technology to encrypt the secret-splitted sets of particles $GH_2^i$, i.e. $\rho _{2i}^{GH_{serve{r_1}}^{k}},\rho _{2i}^{GH_{serve{r_2}}^{k}}, \cdots ,\rho _{2i}^{GH_{serve{r_M}}^{k}}$. After each client has encrypted all their plaintext states, they obtain the ciphertext states $\rho _{2i}^{GH_{serve{r_1}}^{k}}{'},\rho _{2i}^{GH_{serve{r_2}}^{k}}{'}, \cdots ,\rho _{2i}^{GH_{serve{r_M}}^{k}}{'}$. The clients then send the corresponding ciphertext states to the servers sequentially. For example, the ciphertext state $\rho _{2i}^{GH_{serve{r_1}}^{k}}{'}$ is sent to the first server $serve{r_1}$, the ciphertext state $\rho _{2i}^{GH_{serve{r_2}}^{k}}{'}$ is sent to the second server $serve{r_2}$, and so on.

\subsubsection{Evaluation}\label{subsec3.1.4}

In this stage, the servers receive the evaluation circuit $C_q^{'}$ sent by the trusted key center, as well as the ciphertext states sent by the clients. For instance, the first server receives the ciphertext state $\rho _{21}^{GH_{serve{r_1}}^k}{'} \otimes \rho _{22}^{GH_{serve{r_1}}^k}{'} \otimes  \cdots  \otimes \rho _{2i}^{GH_{serve{r_1}}^k}{'} \otimes \rho _{2N}^{GH_{serve{r_1}}^k}{'}$. Subsequently, the servers need to evaluate the received ciphertext states by applying the substituted circuit $C_q^{'}$ to them. After evaluating all the circuits, the servers send the evaluated ciphertext states $\rho _{2i}^{GH_{serve{r_1}}^k}{''},\rho _{2i}^{GH_{serve{r_2}}^k}{''}, \cdots ,\rho _{2i}^{GH_{serve{r_M}}^k}{''}$ back to the corresponding clients.

\subsubsection{Decryption}\label{subsec3.1.5}

In the final stage, the crucial tasks involve the trusted key center completing key updating and the clients completing decryption. The following two steps will be executed sequentially:

(i) \textit {Key Update}: According to the key update functions and the pre-substituted quantum circuit, the trusted key center needs to compute the decryption key ${d_{{k_i}}} = \left( {{a_i}',{b_i}'} \right)$, where ${a_i}',{b_i}' \in {\left\{ {0,1} \right\}^{2n}}$, and then send it to the corresponding clients. Since the decryption key will be used to decrypt data after secret splitting, $Charlie$ needs to generate \textit{MN} decryption keys during the key update process. The key update function are:

(a) If $Gate\left[ g \right] = X$ or $Gate\left[ g \right] = Z$, ${e_{{k_i}}}\left( \omega  \right) = \left( {{a_i}\left( \omega  \right),{b_i}\left( \omega  \right)} \right) \to {d_{{k_i}}}\left( \omega  \right) = \left( {{a_i}\left( \omega  \right)',{b_i}\left( \omega  \right)'} \right) = \left( {{a_i}\left( \omega  \right),{b_i}\left( \omega  \right)} \right)$;

(b) If $Gate\left[ g \right] = H$, ${e_{{k_i}}}\left( \omega  \right) = \left( {{a_i}\left( \omega  \right),{b_i}\left( \omega  \right)} \right) \to {d_{{k_i}}}\left( \omega  \right) = \left( {{a_i}\left( \omega  \right)',{b_i}\left( \omega  \right)'} \right) = \left( {{b_i}\left( \omega  \right),{a_i}\left( \omega  \right)} \right)$;

(c) If $Gate\left[ g \right] = S$, ${e_{{k_i}}}\left( \omega  \right) = \left( {{a_i}\left( \omega  \right),{b_i}\left( \omega  \right)} \right) \to {d_{{k_i}}}\left( \omega  \right) = \left( {{a_i}\left( \omega  \right)',{b_i}\left( \omega  \right)'} \right) = \left( {{a_i}\left( \omega  \right),{a_i}\left( \omega  \right) \oplus {b_i}\left( \omega  \right)} \right)$;

(d) If $Gate\left[ g \right] = CNOT$, ${e_{{k_i}}}\left( {\omega ,\omega '} \right) = \left( {{a_i}\left( \omega  \right),{b_i}\left( \omega  \right),{a_i}\left( {\omega '} \right),{b_i}\left( {\omega '} \right)} \right) \to {d_{{k_i}}}\left( {\omega ,\omega '} \right) = \left( {{a_i}\left( \omega  \right)',{b_i}\left( \omega  \right)',{a_i}\left( {\omega '} \right)',{b_i}\left( {\omega '} \right)'} \right) = \left( {{a_i}\left( \omega  \right),{b_i}\left( \omega  \right) \oplus {b_i}\left( {\omega '} \right),{a_i}\left( \omega  \right) \oplus {a_i}\left( {\omega '} \right),{b_i}\left( {\omega '} \right)} \right)$;

(e) If $Gate\left[ g \right] = T/{T^\dag }$, ${e_{{k_i}}}\left( \omega  \right) = \left( {{a_i}\left( \omega  \right),{b_i}\left( \omega  \right)} \right) \to {d_{{k_i}}}\left( \omega  \right) = \left( {{a_i}\left( \omega  \right)',{b_i}\left( \omega  \right)'} \right) = \left( {{a_i}\left( \omega  \right),{b_i}\left( \omega  \right)} \right)$. \\
where ${e_{{k_i}}}\left( \omega  \right) = \left( {{a_i}\left( \omega  \right),{b_i}\left( \omega  \right)} \right)$ represents the encryption key corresponding to the $\omega \left( {\omega  = 1,2, \cdots ,n} \right)$-th qubit of the encryption key ${e_{{k_i}}} = \left( {{a_i},{b_i}} \right)$. The explanations for $\omega '$ and the decryption key ${d_{{k_i}}}\left( \omega  \right) = \left( {{a_i}{{\left( \omega  \right)}^\prime },{b_i}{{\left( \omega  \right)}^\prime }} \right)$ are the same as those for $\omega $ and the encryption key ${e_{{k_i}}}\left( \omega  \right) = \left( {{a_i}\left( \omega  \right),{b_i}\left( \omega  \right)} \right)$.

It is worth noting that in specific cases, for the key update of the $T/{T^\dag }$-gate, according to the circuit substitution rules in Sect.\ref{subsec2.3}, it is known that the key remains unchanged during the update process. This means that the decryption key is the same as the encryption key.

(ii) \textit {Decryption}: When the client $clien{t_i}\left( {i = 1,2,\, \cdots ,N} \right)$ receives the decryption key and the evaluated ciphertext $\rho _{2i}^{GH_{serve{r_1}}^k}{''},\rho _{2i}^{GH_{serve{r_2}}^k}{''}, \cdots ,\rho _{2i}^{GH_{serve{r_M}}^k}{''}$ from the servers, the client decrypts the ciphertext using the quantum one-time pad (QOTP) technique based on the decryption key ${d_{{k_i}}} = \left( {{a_i}',{b_i}'} \right)$, thereby obtaining the plaintext state after homomorphic evaluation.

\subsection{Adding a server}\label{subsec3.2}

Assume a new server $serve{r_{M + 1}}$ wishes to join the dynamic quantum homomorphic encryption scheme with (M+N)-party proposed in Sect.\ref{subsec3.1}. The scheme proposed in this section can handle the instability issue of the server. Before entering this scheme, clients and servers need to complete a series of operations to join the scheme, forming a dynamic quantum homomorphic encryption scheme with (M+N+1)-party. The scheme with (M+N)-party and the scheme with (M+N+1)-party are largely similar. Here, we primarily analyze the differences. The specific operations to be performed are as follows:

During the secret splitting phase, each client $clien{t_i}\left( {i = 1,2,\, \cdots ,N} \right)$ prepares \textit{n} Bell states from ${\beta _{00}} = {1 \mathord{\left/
 {\vphantom {1 {\sqrt 2 }}} \right.
 \kern-\nulldelimiterspace} {\sqrt 2 }}\left( {\left| {00} \right\rangle  + \left| {11} \right\rangle } \right)$. Then, each client $clien{t_i}\left( {i = 1,2,\, \cdots ,N} \right)$ performs the joint n-particle GHZ measurement on the first particle of each pair of Bell states. According to the theory of entanglement swapping and the properties of GHZ measurements described in Sect.\ref{subsec2.2}, the second particle of each pair of Bell states also becomes entangled into the same GHZ state simultaneously. This process forms two sets of particles, denoted as $GH_1^i = \left\{ {\rho _{1i}^{GH_{serve{r_{M + 1}}}^k}} \right\}$ and $GH_2^i = \left\{ {\rho _{2i}^{GH_{serve{r_{M + 1}}}^k}} \right\}$, where $k = 1,2, \cdots ,n$, $i = 1,2,\, \cdots ,N$, \textit{i} represents the i-th client, ${\rho _{1i}}$ and ${\rho _{2i}}$ represent the first and second particles of the Bell states prepared by the i-th client, respectively. Subsequently, the single-qubit measurement in the X-basis is performed separately on the two sets of particles, yielding measurement outcomes $measur{e_{GH_1^i}} = \left\{ {\rho _{1i}^{MR_{serve{r_{M + 1}}}^k}} \right\}$ and $measur{e_{GH_2^i}} = \left\{ {\rho _{2i}^{MR_{serve{r_{M + 1}}}^k}} \right\}$. Based on the basis chosen during the GHZ measurement process, it is straightforward to determine that the secret the i-th client $clien{t_i}\left( {i = 1,2,\, \cdots ,N} \right)$ wishes to share is

\begin{equation}\label{eq23}
\begin{array}{l}
{K_i}' = \rho _{1i}^{M{R_{serve{r_1}}}} \oplus \rho _{1i}^{M{R_{serve{r_2}}}} \oplus  \cdots  \oplus \rho _{1i}^{M{R_{serve{r_M}}}} \oplus \rho _{1i}^{M{R_{serve{r_{M + 1}}}}}\\
{\rm{    }} = \rho _{2i}^{M{R_{serve{r_1}}}} \oplus \rho _{2i}^{M{R_{serve{r_2}}}} \oplus  \cdots  \oplus \rho _{2i}^{M{R_{serve{r_M}}}} \oplus \rho _{2i}^{M{R_{serve{r_{M + 1}}}}}
\end{array}
\end{equation}

For each \textit{k}, corresponding to each bit of secret data, we have the following formula:

\begin{equation}\label{eq24}
\begin{array}{l}
{K_i}'\left( k \right) = \rho _{1i}^{MR_{serve{r_1}}^k} \oplus \rho _{1i}^{MR_{serve{r_2}}^k} \oplus  \cdots  \oplus \rho _{1i}^{MR_{serve{r_M}}^k} \oplus \rho _{1i}^{MR_{serve{r_{M + 1}}}^k}\\
{\rm{         }} = \rho _{2i}^{MR_{serve{r_1}}^k} \oplus \rho _{2i}^{MR_{serve{r_2}}^k} \oplus  \cdots  \oplus \rho _{2i}^{MR_{serve{r_M}}^k} \oplus \rho _{2i}^{MR_{serve{r_{M + 1}}}^k}
\end{array}
\end{equation}

For each \textit{i}, corresponding to each client, the process remains the same.

In the key generation phase, two tasks still need to be completed: the server participate in distributing keys among multiple clients and the trusted key center, and the trusted key center completes the substitution of quantum circuits. The following two steps will be executed sequentially: During the key distribution process, a difference from before is that due to the additional share of secret data after secret splitting, each client $clien{t_i}\left( {i = 1,2,\, \cdots ,N} \right)$ and the trusted key center need to generate one extra encryption key ${e_{{k_i}}} = {\left( {{a_i},{b_i}} \right)_{serve{r_{M + 1}}}}$ compared to before, where ${a_i},{b_i} \in {\left\{ {0,1} \right\}^{2n}}$. In the quantum circuit substitution process, a difference from before is that the trusted key center also needs to send the substituted circuit $C_q^{'}$ to the new server $serve{r_{M + 1}}$. 

During the encryption process, each client $clien{t_i}\left( {i = 1,2,\, \cdots ,N} \right)$ uses their respective encryption key ${e_{{k_i}}} = {\left( {{a_i},{b_i}} \right)_{serve{r_{M + 1}}}}$ and the QOTP (quantum one-time pad) technique to individually encrypt the set of particles resulting from secret splitting. This operation effectively encrypts $GH_2^i = \left\{ {\rho _{2i}^{GH_{serve{r_{M + 1}}}^k}} \right\}$ to obtain the ciphertext set $\rho _{2i}^{GH_{serve{r_{M + 1}}}^k}{'}$. Subsequently, each client sends their corresponding ciphertext set to the server $serve{r_{M + 1}}$.

Apart from the above operations, all the subsequent steps are the same as those performed by all other servers for their corresponding ciphertext states and $serve{r_{M + 1}}$ must be strictly followed. For clients, they are required to cooperate with other servers as they did previously.

\subsection{Removing a server}\label{subsec3.3}

Now, suppose the server $serve{r_{M}}$ in the dynamic quantum homomorphic encryption scheme with (M+N)-party wishes to withdraw from the scheme. After the server exits the scheme, the dynamic quantum homomorphic encryption scheme with (M+N-1)-party is formed. The scheme with (M+N)-party and the scheme with (M+N-1)-party are largely similar. Here, we will mainly analyze the differences.

In the secret splitting phase, each client $clien{t_i}\left( {i = 1,2,\, \cdots ,N} \right)$ performs the joint n-particle GHZ measurement on the first particle of each pair of Bell states, forming two sets of particles, designated as $GH_1^i = \left\{ {\rho _{1i}^{GH_{serve{r_1}}^k},\rho _{1i}^{GH_{serve{r_2}}^k}, \cdots ,\rho _{1i}^{GH_{serve{r_M}}^k}} \right\}$ and $GH_2^i = \left\{ {\rho _{2i}^{GH_{serve{r_1}}^k},\rho _{2i}^{GH_{serve{r_2}}^k}, \cdots ,\rho _{2i}^{GH_{serve{r_M}}^k}} \right\}$. Subsequently, the single-qubit measurement in the X-basis is performed separately on the two sets of particles, yielding measurement outcomes $measur{e_{GH_1^i}} = \left\{ {\rho _{1i}^{MR_{serve{r_1}}^k},\rho _{1i}^{MR_{serve{r_2}}^k}, \cdots ,\rho _{1i}^{MR_{serve{r_M}}^k}} \right\}$ and $measur{e_{GH_2^i}} = \left\{ {\rho _{2i}^{MR_{serve{r_1}}^k},\rho _{2i}^{MR_{serve{r_2}}^k}, \cdots ,\rho _{2i}^{MR_{serve{r_M}}^k}} \right\}$. At this point, a server wants to exit the scheme, but since the client $clien{t_i}\left( {i = 1,2,\, \cdots ,N} \right)$ has already obtained a set of measurements, it is easy to determine that the secret the i-th client wishes to share is

\begin{equation}\label{eq25}
\begin{array}{l}
{K_i}'' = {K_i} \oplus \rho _{1i}^{M{R_{serve{r_M}}}}\\
{\rm{    }} = {K_i} \oplus \rho _{2i}^{M{R_{serve{r_M}}}}
\end{array}
\end{equation}

For each \textit{k}, corresponding to each bit of secret data, we have the following formula:

\begin{equation}\label{eq26}
\begin{array}{l}
{K_i}''\left( k \right) = {K_i}\left( k \right) \oplus \rho _{1i}^{MR_{serve{r_M}}^k}\\
{\rm{         }} = {K_i}\left( k \right) \oplus \rho _{2i}^{MR_{serve{r_M}}^k}
\end{array}
\end{equation}

For each \textit{i}, corresponding to each client, the process remains the same.

After this, $serve{r_M}$ will be excluded, and no further actions are necessary. The server doesn't even need to be online.

\section{Analysis}\label{sec4}

In this chapter, the correctness, security and efficiency of the proposed scheme will be analyzed. Among them, the correctness is the most critical part, the security can ensure that secret information will not be leaked, and the whole scheme is higher in terms of qubit efficiency than other schemes.

\subsection{Correctness analysis}\label{subsec4.1}

\begin{theorem}\label{thm1}
If the trusted key center, clients and servers can be honestly executed in the secret splitting, key generation, encryption, evaluation, and decryption stages, the clients can get the correct result.
\end{theorem}

\begin{proof}
In the proposed scheme, for the input plaintext $\omega $, the encryption key ${e_k} = \left( {a,b} \right)$, and the quantum circuit ${C_q}$, the homomorphism evaluation should be satisfied

\begin{equation}\label{eq27}
VerDec\left( {Eval_{{\rho _{evk}}}^{{C_q}}\left( {Enc\left( {Split\left( \omega  \right)} \right)} \right)} \right) = {\Phi _{{C_q}}}\left( \omega  \right)
\end{equation}
where $VerDec\left( {Eval_{{\rho _{evk}}}^{{C_q}}\left( {Enc\left( {Split\left( \omega  \right)} \right)} \right)} \right)$ represents the result of the plaintext data after the secret splitting, key generation, encryption, evaluation, and decryption are performed, and ${\Phi _{{C_q}}}\left( \omega  \right)$ represents the result of the quantum circuit ${C_q}$ directly acting on the plaintext data. If the trusted key center, clients, and servers can perform honestly in the secret splitting, key generation, encryption, evaluation, and decryption phases, the result should be equal.

As is well-known, any quantum circuit is composed of Clifford gates and $T$-gate. Suppose the quantum circuit used in the scheme is denoted as ${C_q} = {\Omega _1} \otimes {T_2} \otimes {\Omega _3} \otimes  \cdots  \otimes {\Omega _n}$, acting on \textit{n} qubits represented by $\left| \omega  \right\rangle  = \left| {{\omega _1}} \right\rangle  \otimes \left| {{\omega _2}} \right\rangle  \otimes \left| {{\omega _3}} \right\rangle  \otimes  \cdots  \otimes \left| {{\omega _n}} \right\rangle $. Specifically, ${\Omega _1}$ is applied to the first qubit $\left| {{\omega _1}} \right\rangle $, $T$-gate is applied to the second qubit $\left| {{\omega _2}} \right\rangle $, and so forth, sequentially applied to their respective qubits. Simultaneously, suppose the encryption key is ${e_k} = \left( {a,b} \right)$, where $a,b \in {\left\{ {0,1} \right\}^{2n}}$, and the decryption key is ${d_k} = \left( {a',b'} \right)$, where $a',b' \in {\left\{ {0,1} \right\}^{2n}}$. Throughout the process, the following equation holds true:

\begin{equation}\label{eq28}
\begin{gathered}
  {C_q}{X^a}{Z^b}\left| \omega  \right\rangle  \hfill \\
   = \left( {{\Omega _1} \otimes {T_2} \otimes {\Omega _3} \otimes  \cdots  \otimes {\Omega _n}} \right)\left( {{X^{{a_1}}}{Z^{{b_1}}}\left| {{\omega _1}} \right\rangle  \otimes {X^{{a_2}}}{Z^{{b_2}}}\left| {{\omega _2}} \right\rangle  \otimes {X^{{a_3}}}{Z^{{b_3}}}\left| {{\omega _3}} \right\rangle  \otimes  \cdots  \otimes {X^{{a_n}}}{Z^{{b_n}}}\left| {{\omega _n}} \right\rangle } \right) \hfill \\
   = {\Omega _1}{X^{{a_1}}}{Z^{{b_1}}}\left| {{\omega _1}} \right\rangle  \otimes {R_z}\left( {{{\left( { - 1} \right)}^{{a_2}}}{\pi  \mathord{\left/
 {\vphantom {\pi  4}} \right.
 \kern-\nulldelimiterspace} 4}} \right){X^{{a_2}}}{Z^{{b_2}}}\left| {{\omega _2}} \right\rangle  \otimes {\Omega _3}{X^{{a_3}}}{Z^{{b_3}}}\left| {{\omega _3}} \right\rangle  \otimes  \cdots  \otimes {\Omega _n}{X^{{a_n}}}{Z^{{b_n}}}\left| {{\omega _n}} \right\rangle  \hfill \\
   = {X^{{a_1}'}}{Z^{{b_1}'}}{\Omega _1}\left| {{\omega _1}} \right\rangle  \otimes {X^{{a_2}}}{Z^{{b_2}}}{T_2}\left| {{\omega _2}} \right\rangle  \otimes {X^{{a_3}'}}{Z^{{b_3}'}}{\Omega _3}\left| {{\omega _3}} \right\rangle  \otimes  \cdots  \otimes {X^{{a_n}'}}{Z^{{b_n}'}}{\Omega _n}\left| {{\omega _n}} \right\rangle  \hfill \\
   = {X^{a'}}{Z^{b'}}{C_q}\left| \omega  \right\rangle  \hfill \\ 
\end{gathered}
\end{equation}

From Equation (\ref{eq28}), we know that the decrypted result after performing computations on the ciphertext is equivalent to the result of direct computation on the plaintext, theoretically proving the correctness of the scheme. At the same time, since the quantum circuit can be arbitrary, this also demonstrates that the scheme can achieve quantum homomorphic encryption for any quantum circuit, thus possessing full homomorphism.
\end{proof}

Here, we will conduct a four-qubit simulation experiment on IBM Quantum Experience to further demonstrate the correctness of the proposed scheme. The quantum circuit used in the experiment is ${C_q} = {\left( {XH} \right)_1} \otimes {T_2} \otimes CNO{T_{23}} \otimes {\left( {ZS} \right)_4}$, acting on four qubits denoted by $\left| \omega  \right\rangle  = \left| 1 \right\rangle  \otimes \left| 0 \right\rangle  \otimes \left| 1 \right\rangle  \otimes \left| 0 \right\rangle $, where $CNO{T_{23}}$ indicates the second qubit as the control qubit and the third qubit as the target qubit. Simultaneously, suppose the encryption key is ${e_k} = \left\{ {\left( {0,0} \right),\left( {0,1} \right),\left( {1,0} \right),\left( {1,1} \right)} \right\}$. According to the key update rule, the decryption key ${d_k} = \left\{ {\left( {0,0} \right),\left( {0,1} \right),\left( {1,0} \right),\left( {1,0} \right)} \right\}$ can be computed. Fig.\ref{Fig-3} shows the quantum circuit of the entire quantum homomorphic encryption scheme, in which the plaintext quantum state setting, encryption, homomorphic evaluation, decryption and measurement are separated by long dashed lines. Fig.\ref{Fig-4} shows the results of the quantum circuit diagram of the process.

\begin{figure}[htbp]
\centering		 
\includegraphics[width=0.9\textwidth]{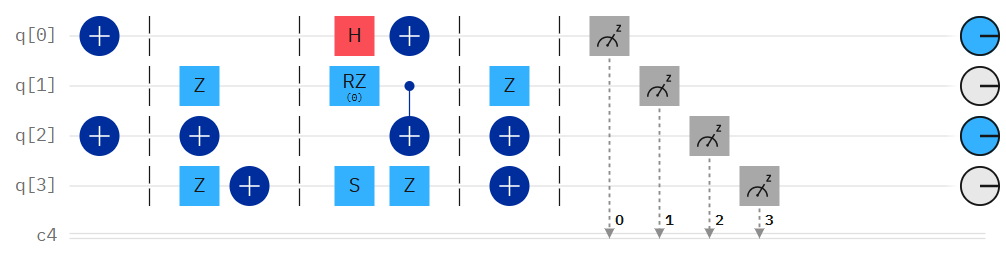}
\caption{\label{Fig-3} The diagram of the quantum circuit}
\end{figure}

\begin{figure}[htbp]
\centering		 
\includegraphics[width=0.9\textwidth]{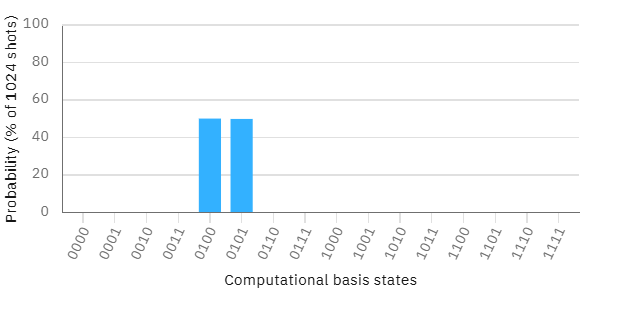}
\caption{\label{Fig-4} The result of the quantum circuit}
\end{figure}

According to theoretical calculations, the final state $\left| {{\omega _{final}}} \right\rangle $ of the four qubits has a 50\% probability of being in state $\left| {0010} \right\rangle $ and a 50\% probability of being in state $\left| {1010} \right\rangle $. As shown in Fig.\ref{Fig-4}, the result of running the quantum circuit also has a 50\% probability of being in state $\left| {0010} \right\rangle $ and a 50\% probability of being in state $\left| {1010} \right\rangle $. In summary, since the two results are identical, this further demonstrates that the proposed scheme is completely correct.

\subsection{Security analysis}\label{subsec4.2}

\begin{theorem}\label{thm2}
The proposed scheme is information-theoretically secure, ensuring that no information about the plaintext data or the key will be leaked.
\end{theorem}

\begin{proof}
 In the proposed scheme, both encryption and decryption use the one-time pad technique \cite{31,32}. The quantum one-time pad (QOTP) technique relies solely on physical mechanisms and does not depend on classical computational hardness assumptions. Suppose the plaintext consists of \textit{n} qubits denoted by $\left| \omega  \right\rangle  = \left| {{\omega _1}} \right\rangle  \otimes \left| {{\omega _2}} \right\rangle  \otimes \left| {{\omega _3}} \right\rangle  \otimes  \cdots  \otimes \left| {{\omega _n}} \right\rangle $, and the ciphertext consists of \textit{n} qubits denoted by $\left| {\omega '} \right\rangle  = \left| {{\omega _1}'} \right\rangle  \otimes \left| {{\omega _2}'} \right\rangle  \otimes \left| {{\omega _3}'} \right\rangle  \otimes  \cdots  \otimes \left| {{\omega _n}'} \right\rangle $. The encryption operator of QOTP is ${X^a}{Z^b}$, and the encryption key $a,b \in {\left\{ {0,1} \right\}^{2n}}$ is randomly generated and used only once. The following equation holds:

\begin{equation}\label{eq29}
\begin{array}{l}
\left| {\omega '} \right\rangle \left\langle {\omega '} \right|\\
 = \left| {{\omega _1}'} \right\rangle \left\langle {{\omega _1}'} \right| \otimes \left| {{\omega _2}'} \right\rangle \left\langle {{\omega _2}'} \right| \otimes \left| {{\omega _3}'} \right\rangle \left\langle {{\omega _3}'} \right| \otimes  \cdots  \otimes \left| {{\omega _n}'} \right\rangle \left\langle {{\omega _n}'} \right|\\
 = {X^{{a_1}}}{Z^{{b_1}}}\left| {{\omega _1}} \right\rangle \left\langle {{\omega _1}} \right|{\left( {{X^{{a_1}}}{Z^{{b_1}}}} \right)^\dag } \otimes {X^{{a_2}}}{Z^{{b_2}}}\left| {{\omega _2}} \right\rangle \left\langle {{\omega _2}} \right|{\left( {{X^{{a_2}}}{Z^{{b_2}}}} \right)^\dag }\\
{\rm{   }} \otimes {X^{{a_3}}}{Z^{{b_3}}}\left| {{\omega _3}} \right\rangle \left\langle {{\omega _3}} \right|{\left( {{X^{{a_3}}}{Z^{{b_3}}}} \right)^\dag } \otimes  \cdots  \otimes {X^{{a_n}}}{Z^{{b_n}}}\left| {{\omega _n}} \right\rangle \left\langle {{\omega _n}} \right|{\left( {{X^{{a_n}}}{Z^{{b_n}}}} \right)^\dag }\\
 = \frac{1}{{{2^{2n}}}}\sum\limits_{a,b \in {{\left\{ {0,1} \right\}}^{2n}}} {{X^a}{Z^b}\left| \omega  \right\rangle \left\langle \omega  \right|{{\left( {{X^a}{Z^b}} \right)}^\dag }} \\
 = \frac{{{I_{{2^n}}}}}{{{2^n}}}
\end{array}
\end{equation}

From Equation (\ref{eq29}), it is evident that when the client encrypts the plaintext state $\left| \omega  \right\rangle $ using QOTP, the resulting ciphertext quantum state $\left| {\omega '} \right\rangle $ is completely mixed. Therefore, the server cannot obtain any information about the plaintext state $\left| \omega  \right\rangle $ or the key ${e_k} = \left( {a,b} \right)$.

Specifically, during the secret splitting phase, the plaintext data remains solely in the possession of the client without any transmission, thus it is not susceptible to any attacks. During the key generation phase, the process of distributing the key can assess the presence of eavesdroppers through the error rate. Moreover, the MDI-QKD protocol has been proven effective in resisting eavesdropping attacks \cite{33}. In the process of quantum circuit substitution, replacing $T/{T^\dag }$-gate does not lead to key leakage, as demonstrated in detail \cite{29}. During the encryption phase, the client uses the QOTP technology to encrypt the plaintext, resulting in ciphertext that is in a completely mixed state, from which no useful information can be obtained. In the evaluation phase, the server receives ciphertext that remains in a completely mixed state. Because the server or any potential eavesdropper lacks any information about the key, they cannot decrypt it correctly. In the decryption phase, the completely trusted key center will send the decryption key to the clients through an authenticated classical channel. Throughout the entire process of the proposed scheme, protection is provided by QOTP.
\end{proof}

Next, we will make a comparison between the similar quantum homomorphic encryption schemes and our proposed scheme. Relevant parameters for comparison include information-theoretic security (ITS), full homomorphism (FH), compactness, non-interactivity (NI), and dynamics. The specific results of the comparison are shown in Table \ref{Table-2}.

\begin{table}[htbp]
\begin{center}
 \caption{Comparison of similar quantum homomorphic encryption schemes}
  \label{Table-2}
  \begin{tabular}{@{}cccccccc@{}}
    \hline
    QHE scheme &  ITS &  FH &  Compactness &  NI &  Dynamics \\ \hline
Broadbent et al. \cite{14}  & No & Yes & Yes & Yes & No\\
Liang et al. \cite{22}  & Yes & Yes & No & Yes & No\\
Chang et al. \cite{25}  & Yes & Yes & No & No & Yes\\
Wang et al. \cite{29}  & Yes & Yes & No & Yes & No\\
Our scheme  & Yes & Yes & No & Yes & Yes\\
    \hline
  \end{tabular}
\end{center}
\end{table}

\subsection{Efficiency analysis}\label{subsec4.3}

In this section, we will analyze the quantum bit efficiency of the proposed scheme and compare it with other similar schemes in Table \ref{Table-3}.

The quantum bit efficiency in the quantum secret sharing scheme is defined as $\eta  = {c \mathord{\left/
 {\vphantom {c q}} \right.
 \kern-\nulldelimiterspace} q}$, where \textit{c} represents the number of bits of the master key to be shared, and \textit{q} represents the number of quantum bits generated during the sharing process.

In the case of multiple clients and multiple servers, the operations of each client are nearly identical, effectively involving parallel operations across multiple clients. For simplicity, we will now analyze the scenario where one client interacts with multiple servers, denoted as the scheme with (M+1)-party. In the scheme, when the client performs secret splitting, they must prepare a Bell state for each server. This totals \textit{M} Bell states, involving 2\textit{M} quantum bits in total. These Bell states collectively encode one quantum bit of data. Therefore, the encoding efficiency of the proposed scheme is $\eta  = \frac{1}{{2M}}$.

\begin{table}[htbp]
\begin{center}
 \caption{Comparison of the efficiency of similar schemes}
  \label{Table-3}
  \begin{tabular}{@{}cccccccc@{}}
    \hline
   Scheme &  Required quantum resources &  Qubit efficiency &  Required Quantum capabilities \\ \hline
Hsu et al. \cite{34}  & Bell state & 1/8 & Bell measurement\\
Yi et al. \cite{35}  & Two-particle entangled state & 2/11 & Z-basis measurement\\
Tsai et al. \cite{36}  & W-state & 1/8 & Z-basis measurement\\
Yang et al. \cite{37}  & Bell state & 1/8 & GHZ measurement\\
Our scheme  & Bell state & 1/4 & GHZ and Z-basis measurement\\
    \hline
  \end{tabular}
\end{center}
\end{table}

\section{Conclusion}\label{sec5}

In this paper, a novel multi-party dynamic quantum homomorphic encryption scheme based on rotation operators is proposed. The proposed scheme not only handles the volatility problem of the servers dynamically, but also eliminates errors caused by homomorphic evaluation of $T$-gate non-interactively, thus we extend this scheme to support a multi-client multi-server mode. In addition, the trusted key center is introduced to be responsible for secure key generation, key updating and circuit substitution for the homomorphic evaluation stage, specifically referring to $T/{T^\dag }$-gate circuit substitution. This reduces client-side operations and lowers the quantum capabilities required from clients. Throughout the entire process, the QOTP ensures that sensitive data remains secure. Simultaneously, we aim for the constructed scheme to be more applicable in quantum distributed networks, accelerating the development of efficient and secure large-scale quantum information transmission. In future work, efforts can focus on leveraging additional quantum properties to optimize processes and enhance the compactness of quantum homomorphic encryption schemes.

\section{Acknowledgement}\label{sec6}

This research was supported in part by the National Natural Science Foundation of China (12161061, 62201525, 61762068), in part by the Beijing Municipal Education Commission Scientific Research Project(KM202310015002, KM202110015003,), in part by the Special Fund for the Local Science and Technology Development of the Central Government under Grant (2020ZY0014), in part by Natural Science Foundation of Inner Mongolia Grant (2021MS01022, 2021MS06011, 2023LHMS06018) and in part by the Youth Excellent Project of Beijing Institute of Graphic Communication(Ea202411).

\bibliographystyle{unsrtnat}
\bibliography{quantumarticle}

\end{document}